\documentclass{elsarticle}
\pdfoutput=1

\usepackage{amsfonts, amsmath, amssymb, amsthm}
\usepackage[T1]{fontenc} %
\usepackage[latin1]{inputenc} %
\usepackage{graphicx}
\usepackage[normalsize]{subfigure}
\usepackage{color, fancybox, calc}
\usepackage{latexsym,amssymb,amsmath}
\usepackage{fullpage}
\usepackage{array,url, appendix}
\usepackage{refcount}

\newtheorem{theorem}{Theorem}
\newtheorem{lemma}[theorem]{Lemma}
\newtheorem{coro}[theorem]{Corollary}

\newtheorem{conj}[theorem]{Conjecture}
\newtheorem{prop}[theorem]{Proposition}
\newtheorem{claim}[theorem]{Claim}
\newtheorem{obs}[theorem]{Observation}

\def\P{\mathcal P}
\def\bp{\mathbf{bp}}
\def\bpor{\mathbf{{bp}_{or}}}

\newcommand{\ceiling}[1]{\left\lceil #1 \right\rceil}
\renewcommand{\O}{{\cal O}}
\newcommand{\F}{\mathcal{F}}
\renewcommand{\L}{\mathcal{L}}
\renewcommand{\S}{\mathcal{S}}
\newcommand{\U}{\mathcal{U}}
\newcommand{\KG}{\mathcal{K}_G}
\newcommand{\SG}{\mathcal{S}_G}
\newcommand{\Gnp}{G(n,p)}

\renewcommand{\O}{{\cal O}}
\renewcommand{\o}{o}

\newcommand{\Col}{\textsc{Col}}

\begin{document}
\title{Clique versus Independent Set}

\author[lirmm]{N. Bousquet}
\ead{nicolas.bousquet@lirmm.fr}

\author[ens]{A. Lagoutte\corref{cor1}}
\ead{aurelie.lagoutte@ens-lyon.fr}

\author[ens]{S. Thomass\'e}
\ead{stephan.thomasse@ens-lyon.fr}

\cortext[cor1]{Corresponding author}

\address[lirmm]{AlGCo project-team, CNRS, LIRMM, 
161 rue Ada, 34392 Montpellier Cedex5 France.}
\address[ens]{LIP, UMR 5668 ENS Lyon - CNRS - UCBL - INRIA, Universit\'e de Lyon, 46, all\'ee de l'Italie, 69364 Lyon France.}

\date{\today}                                 %

\begin{frontmatter}

%
%
%
%
%
%

\begin{abstract} 
Yannakakis' Clique versus Independent Set problem ($CL-IS$) in communication complexity
asks for the minimum number of cuts separating cliques from stable sets in a graph, called CS-separator. Yannakakis provides a quasi-polynomial CS-separator, i.e. 
of size $O(n^{\log n})$,
and addresses the problem of finding a polynomial CS-separator. This question is still open even for perfect graphs. We show that a polynomial CS-separator almost surely exists for random graphs.
Besides, if $H$ is a 
split graph (i.e. has a vertex-partition into a clique and a stable set)
then there exists a constant $c_H$ for which we find a $O(n^{c_H})$ CS-separator on the class 
of $H$-free graphs. This generalizes a result of Yannakakis on comparability graphs. We also provide a $O(n^{c_k})$ CS-separator on the class 
of graphs without induced path of length $k$ and its complement. Observe that on one side, $c_H$ is of order $O(|H| \log |H|)$ resulting from Vapnik-Chervonenkis dimension, and on the other side, $c_k$ is a tower function, due to an application of the regularity lemma.

One of the main reason why Yannakakis' $CL-IS$ problem is fascinating is that it admits equivalent 
formulations. Our main result in this respect is to show that a polynomial CS-separator is equivalent
to the polynomial Alon-Saks-Seymour Conjecture, asserting that if a graph has an edge-partition into 
$k$ complete bipartite graphs, then its chromatic number is polynomially bounded in
terms of $k$. 
We also show that the classical approach to the stubborn problem (arising in CSP)
which consists in covering the set of all 
solutions by $O(n^{\log n})$ instances of 2-SAT is again equivalent to the existence of a polynomial CS-separator.
\end{abstract}

\end{frontmatter}

\section{Introduction} 

The goal of this paper is twofold. First, we focus on the Clique-Stable Set separation problem and provide classes of graphs for which polynomial separators exist. Then we show that this classical problem from communication complexity is equivalent to one in graph theory and one in CSP. Let us make a brief overview of each domain focusing on the problem. 

\begin{paragraph}{Communication complexity and the Clique-Stable Set separation} A \emph{clique} is a complete induced
subgraph and a \emph{stable set} is an induced subgraph with no
edge.
Yannakakis introduced in \cite{Yannakakis91} the following communication complexity problem, called \emph{Clique versus Independent Set} ($CL-IS$ for brevity): given a publicly known graph $\Gamma$ on $n$ vertices, Alice and Bob agree on a protocol, then Alice is given a clique and Bob is given a stable set. They do not know which clique or which stable set was given to the other one, and their goal is to decide whether the clique and the stable set intersect or not, by minimizing the worst-case number of exchanged bits. Note that the intersection of a clique and a stable set is at most one vertex. In the deterministic version, Alice and Bob send alternatively messages one to each other, and the minimization is on the number of bits exchanged between them. It is a long standing open problem to prove a $\O(\log^2 n)$ lower bound for the deterministic communication complexity. In the non-deterministic version, a 
prover knowing the clique and the stable set sends a certificate in order to convince both Alice and Bob of the right answer. Then, Alice and Bob 
exchange one final bit, saying whether they agree or disagree with the certificate. The aim is to minimize the size of the certificate.

In this particular setting, a certificate proving that the clique and the stable set intersect is just the name of the vertex in the intersection. Such a certificate clearly has logarithmic size. Convincing Alice and Bob that the clique and the stable set do not intersect is much more complicated. A certificate can be a bipartition of the vertices such that the whole clique is included in the first part, and the whole stable set is included in the other part. Such a partition is called a cut that separates the clique and the stable set. A family $\F$ of $m$ cuts such that for every disjoint clique and stable set, there is a cut in $\F$ that separates the clique and the stable set is called a CS-separator of size $m$. Observe that Alice and Bob can agree on a CS-separator at the beginning, and then the prover just gives the name of a cut that separates the clique and the stable set: the certificate has size $\log_2 m$. Hence if there is a CS-separator of polynomial size in $n$, one 
can ensure a non-deterministic certificate of size $\O(\log_2 n)$. 

Yannakakis proved that there is a $c \log_2 n$ certificate for the $CL-IS$ problem if and only if there is a CS-separator of size $n^c$. The existence of such a CS-separator is called in the following the Clique-Stable Set separation problem. The best upper bound so far, due to Hajnal (cited in \cite{LovaszSurvey}), is the existence for every graph $G$ of a CS-separator of size $n^{(\log n)/2}$. The $CL-IS$ problem arises from an optimization question which was studied both by Yannakakis \cite{Yannakakis91} and by Lov\'asz \cite{Lovasz94}. The question is to determine if the stable set polytope of a graph is the projection of a polytope in higher dimension, with a polynomial number or facets (called extended formulation). The existence of such a polytope in higher dimension implies the existence of a polynomial CS-separator for the graph. Moreover, Yannakakis proved that the answer is positive for several subclasses of perfect graphs, such as comparability graphs and their complements, chordal graphs and their complements, and Lov\'asz  proved it for a generalization of series-parallel graphs called $t$-perfect graphs. The existence of an extended formulation for general graphs has recently been disproved by Fiorini et al. \cite{
Fiorini11}, and is still open on perfect graphs.

\end{paragraph}

\begin{paragraph}{Graph coloring and the Alon-Saks-Seymour conjecture}
 
Given a graph $G$, the bipartite packing number, denoted by $\bp$, is the minimum number of edge-disjoint complete bipartite graphs needed to partition the edges of $G$. The Alon-Saks-Seymour conjecture (cited in \cite{Kahn94}) states that if a graph has bipartite packing number $k$, then its chromatic number $\chi$ is at most $k+1$. It is inspired from the Graham-Pollak theorem \cite{GrahamP72} which states that $\bp(K_n)=n-1$. Huang and Sudakov proposed in \cite{Sudakov10} a counterexample to the Alon-Saks-Seyour conjecture (then generalized in \cite{CioabaTait12}), twenty-five years after its statement. Actually they proved that there is an infinite family of graphs for which $\chi \geq \Omega(\bp^{6/5})$. The Alon-Saks-Seymour conjecture can now be restated as the \emph{polynomial} Alon-Saks-Seymour conjecture: is the chromatic number polynomially upper bounded in terms of $\bp$? Moreover, Alon and Haviv \cite{AlonH} observed that a gap $\chi \geq \Omega(\bp^c)$ for some graphs would imply a $\Omega(n^c)$ lower bound for the Clique-Stable Set separation problem. Consequently, Huang and Sudakov's result gives a $\Omega(n^{6/5})$ lower bound. This in turns implies a $6/5 \log_2(n) - \O(1)$ lower bound on the non-deterministic communication complexity of $CL-IS$ when the clique and the stable set do not intersect. This lower bound has been improved to $3/2 \log_2(n) - \O(1)$, by Amano \cite{Amano}, using a notion of oriented bipartite packing number, which we also introduced independently. 

A generalization of the bipartite packing number of a graph is the $t$-biclique number, denoted by $\bp_t$. It is the minimum number of complete bipartite graphs needed to cover the edges of the graph such that each edge is covered at least once and at most $t$ times. It was introduced by Alon \cite{Alon97} to model neighborly families of boxes, and the most studied question so far is finding tight bounds for $\bp_t(K_n)$.  
\end{paragraph}

\begin{paragraph}{Constraint satisfaction problem and the stubborn problem}
The complexity of the so-called \emph{list-$M$ partition problem} has been widely studied in the last decades (see \cite{Schell08} for an overview). $M$ stands for a fixed $k\times k$ symmetric matrix filled with $0, 1$ and $\ast $. The input is a graph $G=(V,E)$ together with a list assignment $\L: V \to \mathcal{P}(\{A_1, \ldots, A_k\})$ and the question is to determine whether the vertices of $G$ can be partitioned into $k$ sets $A_1, \ldots, A_k$ respecting two types of requirements. The first one is given by the list assignments, that is to say $v$ can be put in $A_i$ only if $A_i \in \L(v)$. The second one is described in $M$, namely: if $M_{i,i}=0$ (resp. $M_{i,i}=1$), then $A_i$ is a stable set (resp. a clique), and if $M_{i,j}=0$ (resp. $M_{i,j}=1$), then $A_i$ and $A_j$ are completely non-adjacent (resp. completely adjacent). If $M_{i,i}=\ast$ (resp. $M_{i,j}=\ast$), then $A_i$ can be any set (resp. $A_i$ and $A_j$ can have any kind of adjacency).

Feder et al. \cite{FederHH03,FederHKM03} proved a \emph{quasi-dichotomy theorem}. The list-$M$ partition problems are classified between NP-complete and quasi-polynomial time solvable (i.e. time $\O(n^{c \log n})$ where $c$ is a constant). Moreover, many investigations have been made about small matrices $M$ ($k \leq 4$) to get a \emph{dichotomy theorem}, meaning a classification of the list-$M$ partition problems between polynomial time solvable and NP-complete. Cameron et al. \cite{CameronEHS07} reached such a dichotomy for $k \leq 4$, except for one special case (and its complement) then called the \emph{stubborn problem} (
the corresponding symmetric matrix has size 4; $M_{1,1}=M_{2,2}=M_{1,3}=M_{3,1}=0$, $M_{4,4}=1$; the other entries are $\ast$), which remained only quasi-polynomial time solvable. Cygan et al. \cite{Cygan10} closed the question by finding a polynomial time algorithm solving the stubborn problem. More precisely, they found a polynomial time algorithm for \textsc{3-Compatible Coloring}, which was introduced in \cite{FederH06} and said to be no easier 
than the stubborn problem. \textsc{3-Compatible Coloring} has also been introduced and studied in \cite{KostoZ08} under the name \textsc{Adapted List Coloring}, and was proved to be a model for some strong scheduling problems. It is defined in the following way:

%

\textsc{$3$-Compatible Coloring Problem
  ($3$-CCP)}
\\ \textbf{Input:} An edge coloring $f_E$ of the complete
graph on $n$ vertices with $3$ colors $\{A,B,C\}$.
\\ \textbf{Question:} Is there a coloring of the vertices with $\{A,B,C\}$, such that no edge has the same color as both its endpoints?

\end{paragraph}

\begin{paragraph}{Contribution} 
The Clique-Stable Set separation problem will be considered as our reference problem. 
More precisely, we start in Section \ref{sec: CS-sep} by proving that there is a polynomial CS-separator for four classes of graphs: random graphs, split-free graphs, graphs with no induced path $P_k$ on $k$ vertices nor its complement, and graphs with no induced $P_5$. The proof for random graphs is based on random cuts. In the second case, it is based on Vapnik-Chervonenkis dimension. In the third one, it follows the scheme of the proof of the Erd\H{o}s-Hajnal conjecture for graphs with no induced path of length $k$ nor its complement. For graphs with no induced $P_5$, it is a direct consequence of a result of Lokshtanov, Vatshelle, and Villanger \cite{Lok13} used to compute the maximal independent set in such graphs and involving cliques of minimal triangulations.

 In Section \ref{sec: ASS}, we extend Alon and Haviv's observation and prove the equivalence between the polynomial Alon-Saks-Seymour conjecture and the Clique-Stable separation. It follows from an intermediate result, also interesting by itself: for every integer $t$, the chromatic number $\chi$ can be bounded polynomially in terms of $\bp$ if and only if it can be polynomially bounded in terms of $\bp_t$. We also introduce the notion of oriented bipartite packing number, in which the Clique-Stable Set separation exactly translates. For instance, we show that the size of a maximum fooling set of $CL-IS$ gives quite precise information on the oriented bipartite packing number of the complete graph.

In Section \ref{sec: 3-CCP and the stubborn pb}, we highlight links between the Clique-Stable Set separation problem and both the stubborn problem and 3-CCP. The quasi-dichotomy theorem for list-$M$ partitions proceeds by covering all the solutions by $\O(n^{\log n})$ particular instances of 2-SAT, called 2-list assignments. A natural extension would be a covering of all the solutions with a polynomial number of 2-list assignments. We prove that the existence of a polynomial covering of all the maximal solutions (to be defined later) for the stubborn problem is equivalent to the existence of such a covering for all the solutions of 3-CCP, which in turn is equivalent to the $CL-IS$ problem.

\end{paragraph}

\section{Definitions}

Let $G=(V,E)$ be a graph and $k$ be an integer. $V(G)$ is the set of vertices of $G$ and $E(G)$ is its set of edges. An edge $uv \in E$ links its
two \emph{endpoints} $u$ and $v$. The \emph{neighborhood $N_G(x)$ of
$x$} is the set of vertices $y$ such that $xy \in E$. The \emph{closed neighborhood $N_G[x]$ of $x$} is $N_G(x) \cup \{x\}$. The \emph{non-neighborhood $N_G^C[x]$} of $x$ is $V \setminus N_G[x]$. We denote $V \setminus N_G(x)$ by $N_G^C(x)$. When there is no ambiguity about the graph under consideration, we denote by $N(x), N[x], N^C[x], N^C(x)$ the previous definitions. For oriented graphs, $N^+(x)$ (resp. $N^-(x)$) denote the out (resp. in) neighborhood of $x$, i.e. the set of vertices $y$ such that $xy \in E$ (resp. $yx \in E$).
 The subgraph \emph{induced by $X \subseteq V$} denoted by $G[X]$ is the graph with vertex set $X$ and edge set $E
\cap (X \times X)$. A \emph{clique of size $n$} is denoted by $K_n$. Note that a clique and a stable set intersect on
at most one vertex. Two subsets of vertices $X,Y \subseteq V$ are \emph{completely adjacent} if for all $x\in X$, $y \in Y$, $xy\in E$. They are \emph{completely non-adjacent} if there are no edge between them.
A graph $G=(V,E)$ is \emph{split} if
$V=V_1 \cup V_2$ and the subgraph induced by $V_1$ is a clique and the
subgraph induced by $V_2$ is a stable set.  A \emph{vertex-coloring}
(resp. \emph{edge-coloring}) of $G$ with a set \Col \ of $k$ colors is a function
$f_V: V \to \Col$ (resp. $f_E: E \to \Col$).

A graph $G$ is \emph{bipartite} if $V$ can be partitioned into $(U,W)$ such that both $U$ and $W$ are stable sets. Moreover, $G$ is \emph{complete} if $U$ and $W$ are completely adjacent. An \emph{oriented bipartite graph} is a bipartite graph together with an edge orientation such that all the edges go from $U$ to $W$.
A \emph{hypergraph} $H=(V,E)$ is composed of a set of vertices $V$ and a set of \emph{hyperedges} $E\subseteq \mathcal{P}(V)$.

\section{Clique-Stable Set separation conjecture}
\label{sec: CS-sep}

The communication complexity problem $CL-IS$ can be formalized by a function $f: X \times Y \to \{0, 1\}$, 
where $X$ is the set of cliques and $Y$ the set of stable sets of a fixed graph $G$ and 
$f(x,y)=1$ if and only if $x$ and $y$ intersect. It can also be represented by a $|X|\times |Y|$ matrix $M$ 
with $M_{x,y}=f(x,y)$. In the non-deterministic version, Alice is given a clique $x$, Bob is given a stable 
set $y$ and a prover gives to both Alice and Bob a certificate of size $N^b(f)$, where $b\in\{0,1\}$, 
in order to convince them that $f(x,y)=b$. Then, Alice and Bob exchange one final bit, saying whether they agree or disagree with 
the certificate. 

The aim is to minimize $N^b(f)$ in the worst case. When $x$ and $y$ intersect on some vertex $v$, the prover can just 
provide $v$ as a certificate, hence $N^1(f)=O(\log n)$.
The best upper bound so far on $N^0(f)$ is $\O(\log^2(n))$ \cite{Yannakakis91}, which actually is not better 
than the bound on the deterministic communication complexity.

A \emph {combinatorial rectangle} $X'\times Y' \subseteq X\times Y$ is a subset of (possibly non-adjacent) rows $X'$ and columns 
$Y'$ of $M$. It is \emph {$b$-monochromatic} if for all $(x,y) \in X' \times Y'$, $f(x,y)=b$. The minimum 
number of $b$-monochromatic combinatorial rectangles needed to cover the $b$-inputs of $M$ is denoted by $C^b(f)$ and 
verifies $N^b(f)= \ceiling{\log_2 C^b(f)}$ \cite{Kushilevitz99}. 
A \emph{fooling set} is a set $\F$ of $b$-inputs of $M$ such that for all $(x,y), (x',y') \in \F$, $f(x',y)\neq b$ 
or $f(x,y')\neq b$. In other words, a fooling set is a set of $b$-inputs of $M$ that cannot be pairwise contained into 
the same $b$-monochromatic rectangle. Hence, it provides a lower bound on $C^b(f)$.
 Given a 0-monochromatic rectangle $X'\times Y'$, one can construct a partition $(A,B)$ by putting in $A$ every 
vertex appearing in a clique of $X'$, and putting in $B$ every vertex appearing in a stable set of $Y'$. There 
is no conflict doing this since no clique in $X'$ intersects any stable set in $Y'$. We then extend $(A,B)$ into a partition
of the vertices by arbitrarily putting the other vertices into $A$. Observe that $(A,B)$ separates every clique in 
$X'$ from every stable set in $Y'$. Conversely, 
a partition that separates some cliques from some stable sets can be interpreted as a 0-monochromatic 
rectangle. Thus finding $C^0(f)$ (or, equivalently $N^0(f)$) is equivalent to finding the minimum number of cuts 
which separate all the cliques and the stable sets. In particular, there is a $\O(\log n)$ certificate for the 
$CL-IS$ problem if and only if there is a polynomial number of partitions separating all the cliques and the stable sets.

A \emph{cut} is a pair $(A,B)$ such that $A \cup B=V$ and $A \cap B=
\varnothing$. It \emph{separates} a clique $K$ and a stable set $S$
if $K \subseteq A$ and $S \subseteq B$. Note that a clique and a stable set can
be separated if and only if they do not intersect. Let $\KG$ be the set
of cliques of $G$ and $\SG$ be the set of stable sets of $G$. We say that a family $\F$ 
of cuts is a \emph{CS-separator} if for all $(K,S) \in  \KG \times \SG$ which do not intersect, 
there exists a cut in $\F$ that separates $K$ and $S$. While it is generally believed that 
the following question is false, we state it in a positive way:

\begin{conj}\label{cliquestable}
{(Clique-Stable Set separation Conjecture)} There is a polynomial $Q$, such that for every graph $G$ on $n$ vertices, there is a CS-separator of size at most $Q(n)$.
\end{conj}

A first very easy result is that we can only focus on maximal cliques and stable sets.

\begin{prop} \label{clique max}
Conjecture \ref{cliquestable} holds if and only if a polynomial family $\F$ of cuts separates all the maximal (in the sense of inclusion) cliques from the maximal stable sets that do not intersect. 
\end{prop}

\begin{proof}
First note that one direction is direct. Let us prove the other one. Assume $\F$ is a polynomial family that separates all the maximal cliques from the maximal stable sets that do not intersect. Let $Cut_{1,x}$ be the cut $(N[x], N^C[x])$ and $Cut_{2,x}$ be the cut $(N(x), N^C(x))$. Let us prove that $\F'=\F \cup \{Cut_{1,x}\vert x\in V\}\cup \{Cut_{2,x}\vert x\in V\}$ is a CS-separator.

Let $(K,S)$ be a pair of clique and stable set. Extend $K$ and $S$ by adding vertices to get a maximal clique $K'$ and a maximal stable set $S'$. Either $K'$ and $S'$ do not intersect, and there is a cut in $\F$ that separates $K'$ from $S'$ (thus $K$ from $S$). Or $K'$ and $S'$ intersect in $x$ (recall that a clique and a stable set intersect on at most one vertex): if $x \in K$, then $Cut_{1,x}$ separates $K$ from $S$, otherwise $Cut_{2,x}$ does.
\end{proof}

Some classes of graphs have a polynomial CS-separator, this is  for instance the case 
when $\mathcal{C}$ is a class of graphs with a polynomial number of maximal cliques (we just cut every maximal clique 
from the rest of the graph). For example, chordal graphs have a linear number of maximal cliques. A generalization of this is a result of Alekseev \cite{Alekseev}, which asserts that the graphs without induced 
cycle of length four have a quadratic number of maximal cliques.

In this part, we first prove that random graphs have a polynomial CS-separator. Then we focus on classes on graph with a specific forbidden induced graph: more precisely, split-free graphs, graphs with no long paths nor antipaths, and graphs with no path on five vertices. Conjecture \ref{cliquestable} is unlikely to be true in the general case, however we believe it may be true on perfect graphs and more generally in the following setting:

\begin{conj}
Let $H$ be a fixed graph. Then the Clique-Stable Set separation conjecture is true on $H$-free graphs.
\end{conj}

\subsection{Random graphs}\label{subsecrand}

Let $n$ be a positive integer and $p \in [0,1]$ (observe that $p$ can depend on $n$). 
We will work on the Erd\H{o}s-R\'enyi model. The random graph $\Gnp$ is a 
probability space over the set of graphs on the vertex set $\{1, \ldots, n \}$ determined by $\Pr[ij \in E]=p$, 
with these events mutually independent.
A family $\F$ of cuts on a graph $G$ with $n$ vertices is a \emph{complete $(a,b)$-separator} if 
for every pair $(A,B)$ of disjoint subsets of vertices with $|A|\leq a$, $|B| \leq b$, there exists a cut 
$(U, V \setminus U)\in \F$ separating $A$ and $B$, namely $A \subseteq U$ and $B \subseteq V \setminus U$.

\begin{theorem} \label{random}
Let $n\in \mathbb{N}$ and $p\in [0,1]$ be the probability of an edge (possibly depending on $n$). Then there exists a family $\F_{n,p}$ of size $\mathcal{O}(n^{7})$ such that for every graph $G\in G(n,p)$, the probability that $\F_{n,p}$ is a CS-separator for $G$ tends to 1 when $n$ goes to infinity.
\end{theorem}

\begin{proof}
We distinguish two cases: assume first that $p\leq 1/\sqrt{n}$  and  consider the probability to have a clique of size $6$:
$$\mathbb{P}(\exists \textrm{ a clique of size } 6)\leq{n \choose 6} p^{{6 \choose 2}}
																				\leq \frac{n^6}{(\sqrt{n})^{15}}
																				 \leq  n^{-3/2}\underset{n\to+\infty}{\longrightarrow} 0$$
Hence every potential clique have size at most five. Define the family $\F_{n,p}$ of size $\mathcal{O}(n^5)$:
\[\F_{n,p}=\{(U,V \setminus U) | \ U\subseteq V, |U|\leq 5\} \ , \]
then the statement holds with $\F_{n,p}$. If $1-p\leq 1/\sqrt{n}$, the proof is the same by exchanging clique and stable set and taking 
\[\F_{n,p}=\{(V \setminus U, U) | \ U\subseteq V, |U|\leq 5\} \ .\]

\noindent For the second case, we can now suppose that both $p>1/\sqrt{n}$ and $1-p>1/\sqrt{n}$. In the following, 
$\log$ denotes the logarithm to base 2. 
Following classical results \cite{bollobas01}  for the case where $p$ is fixed and independent from $n$, let 
$$r_\omega= \frac{3\log n}{- \log p} \qquad \textrm{and} \qquad r_\alpha= \frac{3\log n}{- \log (1-p)} \ .$$
The first goal is to construct a complete $(r_\omega, r_\alpha)$-separator.
Draw a random partition $(V_1, V_2)$ where each vertex is put in $V_1$ independently from the others with 
probability $p$, and put in $V_2$ otherwise. Let $A$ and $B$ be two disjoint subsets of vertices of respective size $r_\omega$ and $r_\alpha$. There are at most $4^n$ such pairs. The probability that $A \subseteq V_1$ and $B \subseteq V_2$ 
is at least $p^{r_\omega} (1-p)^{r_\alpha}$. Observe that $p^{r_\omega} (1-p)^{r_\alpha} = 1/n^{6}$. Then on average $(A,B)$ is 
separated by at least $1/n^{6}$ of all the partitions. By double counting, there exists a partition that separates 
at least $1/n^{6}$ of all the pairs. We delete these separated pairs and add the partition to $\F_{n,p}$, and there remain at most $(1-1/n^{6}) \cdot 4^n$ 
pairs. The same probability for a pair $(A,B)$ to be cut by a random partition still holds, hence we can iterate 
the process $i$ times until $(1-1/n^{6})^i\cdot 4^n \leq 1$. This is satisfied for $i= 2 n^{7}$. Thus starting from $\F_{n,p}=\emptyset$ and adding one by one the selected cuts, we achieve a complete $(r_\omega, r_\alpha)$-separator of size $\mathcal{O}(n^{7})$.

The second goal is to prove that the probability that $\F_{n,p}$ is a CS-separator for $G$ tends to 1 when $n$ goes to infinity. It is enough to prove that the probability that there exists a clique (resp. stable set) of size  $r_\omega$ (resp. $ r_\alpha$ ) tends to 0 when $n$ goes to infinity. Both are similar by exchanging $p$ (resp. clique) and $1-p$ (resp. stable set).
Observe that 
$$\mathbb{P}(\exists K, |K|={r_\omega}, K \textrm{ is a clique})\leq {n \choose {r_\omega}} p^{{r_\omega} \choose 2}$$
Standard calculation using the Stirling approximation shows that this expression is equivalent to $(2\pi)^{-1/2}f(n)$ where

$$f(n)=\left( 1-\frac{{r_\omega}}{n}\right)^{-n-1/2}\left(\frac{n}{{r_\omega}}-1\right)^{{r_\omega}} {r_\omega}^{-1/2}p^{\frac{{r_\omega}({r_\omega}-1)}{2}}$$

\noindent Observe now that, since $1-p>1/\sqrt{n}$, then 
$$\frac{{r_\omega}}{n}\leq \frac{3\log n}{\sqrt{n}+o(\sqrt{n})}\underset{n\to+\infty}{\longrightarrow} 0 \qquad \textrm{thus} \qquad -\left(n+\frac{1}{2}\right)\log\left(1-\frac{{r_\omega}}{n}\right)={r_\omega}+o({r_\omega})$$
Then standard calculation gives
$$ \log(f(n)) \leq  r_{\omega}+ o(r_{\omega})+ {r_\omega}\log n -{r_\omega}\log {r_\omega} -\frac{1}{2}\log r_\omega +\frac{{r_\omega}({r_\omega}-1)}{2}\log p$$
and $$\frac{{r_\omega}({r_\omega}-1)}{2}\log p= -\frac{3}{2} r_\omega \log n + \frac{3}{2}	\log n$$

\noindent  Moreover, since $p>1/\sqrt{n}$, then $ r_{\omega} \geq 6$ so
  $(r_{\omega}+1/2)\log r_{\omega}\geq 0$ and 
  $r_{\omega} \log n \underset{n\to+\infty}{\longrightarrow} +\infty$.  Thus 
 
\begin{align*}
\log(f(n)) & \leq r_{\omega}+ o(r_{\omega})+ {r_\omega}\log n  -\frac{3}{2} r_\omega \log n + \frac{3}{2}	\log n\\
 & \leq - \left(\frac{1}{2}+o(1)\right) r_{\omega} \log n + \frac{3}{2}\log n\\
 & \leq - \left(\frac{1}{2}+o(1)\right) 6 \log n + \frac{3}{2}\log n\\
 & \leq \left(-\frac{3}{2}+o(1)\right) \log n \underset{n\to+\infty}{\longrightarrow} -\infty \ .
 \end{align*}
\end{proof}

Note here that no optimization was made on the degree of the polynomial. Replacing the constant $3$ by $(5/2+\varepsilon)$ in the definition of $r_\omega$ and $r_\alpha$ leads to a CS-separator of size $\mathcal{O}(n^{6+2\varepsilon})$. Moreover, an interesting question would be a lower bound on the degree of the polynomial needed to separate the cliques and the stable sets in random graphs, in particular for the special case $p=1/2$.

\subsection{The case of split-free graphs.}

A graph $\Gamma$ is called \emph{split} if its vertices can be partitioned into a clique and a stable set. A graph $G=(V,E)$ \emph{has an induced $\Gamma$} if there exists $X \subseteq V$ such that the induced graph $G[X]$ is isomorphic to $\Gamma$.
We denote by $\mathcal{C}_\Gamma$ the class of graphs with no induced $\Gamma$. For instance, if 
$\Gamma$ is a triangle with three pending edges, 
 then $\mathcal{C}_\Gamma$ 
contains the class of comparability graphs, for which Lov\'asz showed \cite{Lovasz94} the existence 
of a CS-separator of size $\O(n^2)$. Our goal in this part is to prove that $\mathcal{C}_\Gamma$ has 
a polynomial CS-separator when $\Gamma$ is a split graph.


Let us first state some definitions concerning hypergraphs and VC-dimension. Let $H=(V,E)$ be a hypergraph. The \emph{transversality} $\tau(H)$ is the minimum cardinality of a subset of vertices intersecting each hyperedge. The transversality corresponds to an optimal solution of the following integer linear program: 

\begin{center}
\begin{tabular}{rl}
Minimize: & $\displaystyle{\sum_{x \in V} w(x)}$ \\
\\

Subject to: & 
\begin{tabular}[t]{l}
$\forall x \in V$, $w(x) \in \{ 0,1 \}$\\
 $\forall e \in E$, $\sum_{x \in e} w(x) \geq 1$
\end{tabular}
\end{tabular}
\end{center}

The \emph{fractional transversality} $\tau^*$ is the fractional relaxation of the above linear program. The first condition is then replaced by: for all $x \in V$, $ w(x)\geq 0$.
The \emph{Vapnik-Chervonenkis dimension} or \emph{VC-dimension} \cite{Vapnik} of a hypergraph $H = (V , E)$
is the maximum cardinality of a set of vertices $A \subseteq V$ such that for every $B \subseteq A$ there is an edge
$e \in E$ so that $e \cap A = B$.
The following bound due to Haussler and Welzl \cite{Haussler86} links the transversality, the VC-dimension and the fractional transversality.

\begin{lemma}\label{vchaussler}
 Every hypergraph $H$ with VC-dimension $d$ satisfies
$$ \tau(H) \leq 16  d  \tau^*(H) \log(d \tau^*(H)).$$
\end{lemma}

\begin{theorem}\label{Hfree}
 Let $\Gamma$ be a fixed split graph. Then the Clique-Stable Set conjecture is verified on  $\mathcal{C}_\Gamma$.
\end{theorem}

\begin{proof}
The vertices of $\Gamma$ are partitioned into $(V_1, V_2)$ where $V_1$ is a clique and $V_2$ is a stable set. 
Let $\varphi=\max(|V_1|, |V_2|)$ and $t=64 \varphi (\log(\varphi)+2)$.
Let $G=(V,E) \in \mathcal{C}_\Gamma$ and $\F$ be the following family of cuts. For every clique 
$\{x_1, \ldots, x_r\}$ with $r\leq t$, we note $U= \cap_{1 \leq i \leq r} N[x_i]$
and put $(U, V \setminus U)$ in $\F$. Similarly, for every stable 
set $\{x_1, \ldots, x_r\}$ with $r\leq t$, we note $U=  \cup_{1 \leq i \leq r} N(x_i)$
and put $(U, V \setminus U)$ in $\F$. Since each member of $\F$ is defined with a set of at most $t$ vertices, the size of 
$\F$ is at most $\O(n^t)$. Let us now prove that $\F$ is a CS-separator.
Let $(K,S)$ be a pair of maximal clique and stable set. 
We build $H$ a hypergraph with vertex set $K$. For all $x \in S$, build the hyperedge $K \setminus N_G(x)$ (see Fig.~\ref{fig B'}). Symmetrically, build $H'$ a hypergraph with vertex set $S$. For all $x \in K$, build the hyperedge $S \cap N_G(x)$. The goal is to prove thanks to Lemma \ref{vchaussler} that $H$ or $H'$ has bounded transversality $\tau$. This will enable us to prove that $(C,S)$ is separated by $\F$.

To begin with, let us introduce an auxiliary oriented graph $B$ with vertex set $K \cup S$. For all $x \in K$ and $ y \in S$, put the arc $xy$ if $xy \in E$, and put the arc $yx$ otherwise (see Fig.~\ref{fig B}). For a weight function $w: V\to \mathbb{R}^+$ and a subset of vertices $T\subseteq V$, we define $w(T)=\sum_{x\in T} w(x)$.

\begin{figure}
\centering
\subfigure[A clique $K$ and a stable $S$ in $G$.]{ \label{fig clique_stable1}
 \includegraphics[scale=0.5]{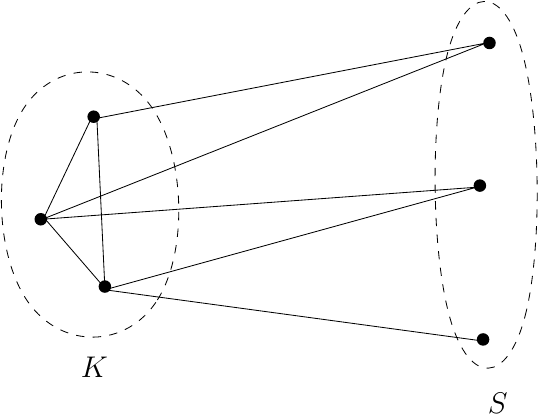}
 }
\subfigure[Hypergraph $H$ where hyperedges are built from the non-neighborhood of vertices from $S$.]{ \hspace{1cm}\label{fig B'}
 \includegraphics[scale=0.5]{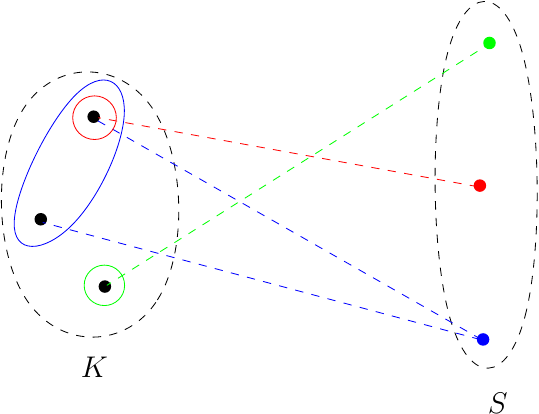}
 }
\subfigure[Graph $B$ built from $K$ and $S$. 
]{ \label{fig B} 
 \includegraphics[scale=0.5]{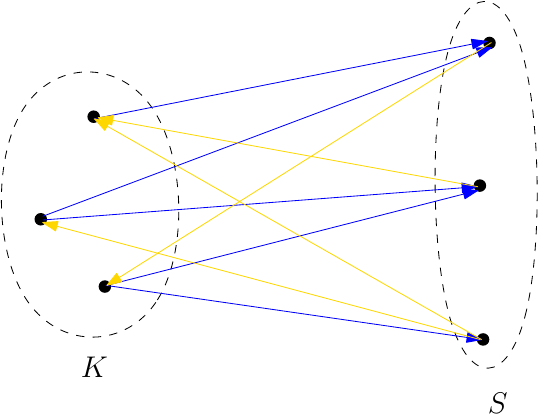} 
 }

\caption{Illustration of the proof of Theorem \ref{Hfree}. For more visibility in \ref{fig B}, forward arcs are drawn in blue and backward arcs in yellow.}
\end{figure}

\begin{lemma} \label{weight}
 In $B$, there exists:
\begin{enumerate}
 \item[(i)] either a weight function $w: K \to \mathbb{R}^+$ such that $w(K)=2$ and $\forall x \in S, w(N^+(x)) \geq 1$.
 \item[(ii)] or a weight function $w: S \to \mathbb{R}^+$ such that $w(S)=2$ and $\forall x \in K, w(N^+(x)) \geq 1$.
\end{enumerate}
\end{lemma}

In the following, let assume we are in case (i) and let us prove that $H$ has bounded transversality. Case (ii) is handled symmetrically by switching $H$ and $H'$.

\begin{lemma} \label{tau star}
 The hypergraph $H$ has fractional transversality $\tau^* \leq 2$.
\end{lemma}

\begin{lemma} \label{VC-dim}
$H$ has VC-dimension bounded by $2\varphi-1$.
\end{lemma}

\noindent Applying Lemmas \ref{vchaussler}, \ref{tau star} and \ref{VC-dim} to $H$, we obtain

\[ \tau(H) \leq 16  d  \tau^*(H) \log(d \tau^*(H)) \leq 64 \varphi (\log(\varphi)+2)=t.\]

\noindent Hence $\tau$ is bounded by $t$ which only depends on $H$.
There must be $x_1, \ldots, x_\tau \in K$ such that each hyperedge of $H$ contains at least one $x_i$.
Consequently, $S \subseteq \cup_{1\leq i \leq t} N^C_G[x_i]$.
Moreover, $K \subseteq (\cap_{1\leq i \leq t} N_G[x_i])=U$ since $x_1, \ldots, x_\tau$ are in the same clique $K$. This means that the cut $(U, V\setminus U) \in \F$ built from the clique $x_1, \ldots, x_\tau$ separates $K$ and $S$.

When case (ii) of Claim \ref{weight} occurs, $H'$ has bounded transversality, so there are $\tau$ vertices  $x_1, \ldots, x_\tau \in S$ such that for all 
$y \in K$, there exists $x_i \in N(y)$. Thus $K \subseteq (\cup_{1\leq i \leq t} N_G(x_i))=U$ and $S \subseteq \cap_{1\leq i \leq t} N^C_G(x_i)$. The cut $(U, V\setminus U) \in \F$ built from the stable set $x_1, \ldots, x_\tau$ separates $K$ and $S$.
\end{proof}

\begin{proof}[Proof of Lemma \ref{weight}] If $x=(x_1, \ldots , x_n) \in \mathbb{R}^n$, we note $x\neq 0$ if there exists $i$ such that $x_i\neq 0$ and we note $x\geq 0$ if for every $i$, $x_i\geq 0$. We use the following variant of the geometric Hahn-Banach separation theorem, from which we derive Claim \ref{farkas lemma adapted}:

\begin{claim}\label{Farkas lemma}
 Let $A$ be a $n\times m$ matrix. Then at least one of the following holds:
\begin{enumerate}
 \item There exists $w \in \mathbb{R}^m$ such that $w\geq 0$, $w\neq 0$ and $Aw\geq 0$.
 \item[or 2.] There exists $y\in \mathbb{R}^n$ such that $y\geq0$, $y \neq 0$ and ${}^t y A\leq0$.
\end{enumerate}

\end{claim}

\begin{proof}
 Call $P \subseteq \mathbb{R}^n$ the convex set composed of all vectors with only positive coordinates. Call $a_1, \ldots, a_m$ the columns vectors of $A$ and $A_{vec}=\{\lambda_1 a_1 + \ldots+ \lambda_m a_m \mid \lambda_1, \ldots, \lambda_m \in \mathbb{R}^+\}$. If $P\cap A_{vec}\neq \{0\}$, then there exists $w \in \mathbb{R}^m$ fulfilling the requirements of the first item. Otherwise, the interior of $P$ and the interior of $A_{vec}$ are disjoint and, according to the geometric Hahn-Banach separation theorem, there is a hyperplane separating them. Call its normal vector on the positive side $y\in \mathbb{R}^n$, then $y$ fulfills the requirements of the second item.
\end{proof}

\begin{claim}\label{farkas lemma adapted}
 For all oriented graph $G=(V,E)$, there exists a weight function $w: V \to [0,1]$ such that $w(V)=1$ and for each vertex $x$, $w(N^+(x)) \geq w(N^-(x))$.
\end{claim}

\begin{proof} 
Let $A$ be the adjacency matrix of the oriented graph $G$, that is to say that $A_{x,y}=1$ if $xy\in E$, $-1$ if $yx \in E$, and 0 otherwise. Apply Lemma \ref{Farkas lemma} to $A$. Either case one occurs and then $w$ is a nonnegative weight function on the columns of $A$, with at least one non zero weight. Moreover, $Aw\geq 0$ so we get $w(N^+(x)) \geq w(N^-(x))$ for all $x\in V$. We conclude by rescaling the weight function with a factor $1/w(V)$.

Otherwise, case two occurs and there is $y\in \mathbb{R}^{n}$ with $y \neq 0$ such that ${}^t y A\leq0$. We get by transposition ${}^tA y \leq 0$ thus $-Ay \leq 0$ since $A$ is an antisymetric matrix, and then $A y \geq 0$. We conclude as in the previous case.
\end{proof}
Apply Claim \ref{farkas lemma adapted} to $B$ to obtain a weight function $w': V \to [0,1]$. Then $w'(V)=1$, so either $w'(K)>0$ or $w'(S)>0$.
Assume $w'(K)>0$ (the other case is handled symmetrically). Consider the new weight function $w$ defined by $w(x)=2w'(x)/w'(K)$ if $x \in K$, and 0 otherwise. Then for all $x \in S$,  on one hand $w(N^+(x)) \geq w(N^-(x))$ by extension of the property of $w'$, and on the other hand, $N^+(x) \cup N^-(x) = K$ by construction of $B$. Thus $w(N^+(x)) \geq w(K)/2 =1$ since $w(K)=2$.
\end{proof}

\begin{proof}[Proof of Lemma \ref{tau star}]
Let us prove that the weight function $w$ given by Lemma \ref{weight} provides a solution to the fractional transversality linear program. Let $e$ be a hyperedge built from the non-neighborhood of $x \in S$. Recall that this non-neighborhood is precisely $N^+(x)$ in $B$, then we have:
$$ \sum_{y \in e} w(y) = w(N^+(x)) \geq 1.$$
Thus $w$ satisfies the constraints of the fractional transversality, and $w(K)\leq 2$, i.e. $\tau^* \leq 2$. 
\end{proof}

\begin{proof}[Proof of Lemma \ref{VC-dim}]
Assume there is a set $A=\{u_1, \ldots, u_\varphi, v_1, \ldots, v_\varphi\}$ of $2\varphi$ vertices of $H$ such that for every $B \subseteq A$ there is an edge $e\in E$ so that $e \cap A=B$. The aim is to exploit the shattering to find an induced $\Gamma$, which builds a contradiction. Recall that the forbidden split graph $\Gamma$ is the union of a clique $V_1=\{x_1, \ldots, x_r\}$ and a stable set $V_2=\{y_1, \ldots, y_{r'}\}$ (with $r, r' \leq \varphi$). Let $x_i \in V_1$, let $\{y_{i_1}, \ldots, y_{i_k}\}=N_\Gamma(x_i) \cap V_2$ be the set of its neighbors in $V_2$.

Consider $\U_i=\{u_{i_1}, \ldots, u_{i_k}\}\cup \{v_i\}$ (possible because $|V_1|, |V_2|\leq \varphi$). By assumption on $A$, there exists $e\in E$ such that $e\cap  A=A\setminus \U_i$. Let $s_i\in S$ be the vertex whose non-neighborhood corresponds to the edge $e$, then the neighborhood of $s_i$ in $A$ is exactly $\U_i$. Let $\U=\{u_1, \ldots, u_\varphi\}$.
Now, forget about the existence of $v_1, \ldots, v_\varphi$, and observe that $N_G(s_i)\cap \U=\{u_{i_1}, \ldots , u_{i_k}\}$. Then $G[\{s_1, \ldots , s_r\} \cup \U]$ is an induced $\Gamma$, which is a contradiction.
\end{proof}

Note that the presence of $v_1, \ldots, v_\varphi$ is useful in case where two vertices of $V_1$ are twins with respect to $V_2$, meaning that their neighborhoods restricted to $V_2$ are the same, call it $N$. Then, $A$ does not ensure that there exist two hyperedges intersecting $A$ in exactly $N$. So the vertices $v_1, \ldots, v_\varphi$ ensure that for two distinct vertices $x_i, x_j$ of $V_1$, the sets $\U_i$ and $\U_j$ are different.
In fact, only $v_1, \ldots, v_{\log \varphi}$ are needed to make $\U_i$ and $\U_j$ distinct: for $x_i \in V_1$, code $i$ in binary over $\log\varphi$ bits and define $\U_i$ to be the union of $\{u_{i_1}, \ldots, u_{i_k}\}$ with the set of $v_{j}$ such that the $j$-th bit is one. 
Thus the VC-dimension of $H$ is bounded by $\varphi + \log \varphi$.

\subsection{The case of $P_k,\overline{P_k}$-free graphs}

The graph $P_k$ is the path with $k$ vertices, and the graph $\overline{P_k}$ is its complement. Let $\mathcal{C}_k$ be the class of graphs with no induced $P_k$ nor $\overline{P_k}$. We prove the following:

\begin{theorem} \label{th: path}
Let $k>0$. The Clique-Stable set conjecture is verified on $\mathcal{C}_k$.
\end{theorem}

The proof relies on this very recent result about $\mathcal{C}_k$, which appears in the study of the Erd\H{o}s-Hajnal property on this class:

\begin{theorem} \label{path} \cite{Thomasse13}
For every $k$, there is a constant $t_k>0$, such that every graph $G\in \mathcal{C}_k$ contains two subsets of vertices $V_1$ and $V_2$, each of size at least $t_k \cdot n$, such that $V_1$ and $V_2$ are completely adjacent or completely non-adjacent. 
\end{theorem}

\begin{proof}[Proof of Theorem \ref{th: path}]
The goal is to prove that every graph in $\mathcal{C}_k$ admits a CS-separator of size $n^c$ with $c=(-1/\log_2(1- t_k))$. We proceed by contradiction and assume that $G$ is a minimal counter-example. Free to exchange $G$ and its complement, by Theorem \ref{path}, there exists two subsets $V_1, V_2$ completely non adjacent, and $|V_1|, |V_2| \geq t_k \cdot n$ for some constant $0<t_k<1$. Call $V_3=V\setminus (V_1\cup  V_2)$. By minimality of $G$,  $G[V_1\cup V_3]$ admits a CS-separator $F_1$ of size $(|V_1|+|V_3|)^c$, and $G[V_2\cup V_3]$ admits a CS-separator $F_2$ of size $(|V_2|+|V_3|)^c$. Let us build $F$ aiming at being a CS-separator for $G$. For every cut $(U,W)$ in $F_1$, build the cut $(U, W\cup V_2)$, and similarly for every cut $(U,W) $ in $F_2$, build the cut $(U, W\cup V_1)$. We show that $F$ is indeed a CS-separator: let $(K,S)$ be a pair of clique and stable set of $G$ that do not intersect, then either $K\subseteq V_1\cup V_3$, or $K\subseteq V_2\cup V_3$ since there is no edge between $V_1$ and $V_2$. By symmetry, suppose $K\subseteq V_1\cup V_3$, then there exists a cut $(U,W)$ in $F_1$ that separates $(K, S\cap (V_1\cup V_3))$ and the corresponding cut $(U, W\cup V_2)$ in $F$ separates $(K,S)$. Finally, $F$ has size at most $2\cdot ((1-t_k)n)^c\leq n^c$.
\end{proof}
%
%


\subsection{The case of $P_5$-free graphs}

When excluding only a path, we can obtain a CS-separator for $P_5$-free graphs thanks to the following result due to Lokshtanov, Vatshelle, and Villanger \cite{Lok13}. A \emph{triangulation} of a graph $G=(V,E)$ is a graph $H=(V,E\uplus F)$ (obtained from $G$ by adding a set of edges $F$ called \emph{fill edges}) such that every cycle of length at least four has a \emph{chord}, that is an edge between two non-consecutive vertices of the cycle. It is a \emph{minimal} triangulation if $H'=(V,E\uplus F')$ is not a triangulation for every $F'\subsetneq F$.

\begin{theorem}[rephrased from \cite{Lok13}] \label{P5-free}
Every $P_5$-free graph $G=(V,E)$ has
a family $\Pi$ of subsets of $V$ with size at most $3n^7$, such that for every maximal stable set $S$ of G with $|S| \geq 2$
there exists a minimal triangulation $H$ of $G$ such that every maximal clique of $H$ is in $\Pi$ and every fill edge has both extremities in $V\setminus S$.
\end{theorem}

\begin{coro}
For every $P_5$-free graph $G$, there exists an $\mathcal{O}(n^8)$ CS-separator.
\end{coro}

\begin{proof}
Let $\Pi$ be the family output by the algorithm of Theorem \ref{P5-free}. Define $F=\Pi\uplus\Pi'\uplus F_0$ where
\begin{align*}
\Pi'&=\{U\setminus\{x\},V\setminus(U\setminus\{x\}) | \ U\in \Pi, x\in V \} \\
F_0&= \{V\setminus\{x\},\{x\})| \ x\in V \} 
\end{align*}
Let $K$ and $S$ be respectively a clique and a stable set of $G$ which do not intersect. If $|S|=1$, $F_0$ separates $K$ from $S$. Otherwise, by property of $\Pi$, there exists a minimal triangulation $H$ of $G$ such that every maximal clique of $H$ is in $\Pi$ and every fill edge has both extremities in $V\setminus S$, in particular $S$ is still a stable set in $H$. Let $K'$ be a maximal clique of $H$ such that $K\subseteq K'$. Then $|K'\cap S| \leq 1$ and $K'\in \Pi$. In particular $(K'\setminus (K'\cap S), (V\setminus K')\cup (K'\cap S))\in F$ separates $K$ and $S$.
\end{proof}

It can be noted that the family $\Pi$ can be efficiently constructed.

\section{Bipartite packing number and graph coloring}

\label{sec: ASS}

The aim of this section is to prove that the polynomial Alon-Saks-Seymour conjecture is equivalent to the 
Clique-Stable Set separation conjecture. We need for this an intermediate step using a new version of the 
Alon-Saks-Seymour conjecture, called the Oriented Alon-Saks-Seymour conjecture. 

\subsection{Oriented Alon-Saks-Seymour conjecture}
Given a graph $G$, the \emph{chromatic number $\chi(G)$} of $G$ is the minimum number of colors needed to color the vertices such that any two adjacent vertices have different colors. The \emph{bipartite packing number $\bp(G)$} of a graph $G$ is the minimum number of edge-disjoint complete bipartite graphs needed to partition the edges of $G$. Alon, Saks and Seymour conjectured that if $\bp(G) \leq k$, then $\chi(G) \leq k+1$.
The conjecture holds for complete graphs. Indeed, Graham and Pollak \cite{GrahamP72} proved that $n-1$ edge-disjoint complete bipartite graphs are needed to partition the edges of $K_n$. A beautiful algebraic proof of this theorem is due to Tverberg \cite{Tverberg82}. The conjecture was disproved by Huang and Sudakov in \cite{Sudakov10} who proved that $\chi \geq \Omega(k^{6/5})$ for some graphs using a construction based on Razborov's graphs \cite{Razborov92}. Huang and Sudakov even conjectured the existence of a graph $G$ such that $\chi(G) \geq 2^{c \log^2(\bp(G))}$ for some constant $c>0$, nevertheless the existence of a polynomial bound is still open.
\begin{conj}{(Polynomial Alon-Saks-Seymour Conjecture)}
There exists a polynomial $P$ such that for every $G$, $\chi(G)\leq P(\bp(G))$.
\end{conj}
We introduce a variant of the bipartite packing number which may lead to a new superlinear lower bound on the Clique-Stable separation. Note that it is the same notion as \emph{ordered biclique covering}, denoted by $\bp_{1.5}$, which has been independently introduced in \cite{Amano}.
The \emph{oriented bipartite packing number $\bpor(G)$} of a non-oriented graph $G$ is the minimum number of oriented complete bipartite graphs such that each edge is covered by an arc in at least one direction (it can be in both directions), but it cannot be covered twice in the same direction (see Fig.~\ref{bpor} for an example).
 A \emph{packing certificate of size $k$} is a set $\{(A_1, B_1), \ldots , (A_k, B_k)\}$ of $k$ oriented bipartite subgraphs of $G$ that fulfill the above conditions restated as follows: for each edge $xy$ of $G$, free to exchange $x$ and $y$, there exists $i$ such that $x\in A_i, y\in B_i$, but there do not exist distinct $i$ and $j$ such that $x\in A_i\cap A_j$ and $y \in B_i\cap B_j$. 
 
 \begin{figure}
 \centering
 \includegraphics[scale=1]{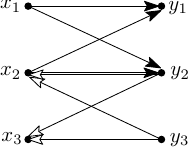}
 \caption{A graph $G$ such that $\bpor(G)=2$ (and $\bp(G)=3$). Two different kinds of arrows show a packing certificate of size 2:$(\{x_1, x_2\}, \{y_1, y_2\})$ and $(\{y_2, y_3\},\{x_2, x_3\})$. The edge $x_2y_2$ is covered once in each direction, while the other edges are covered in exactly one direction. }
 \label{bpor}
 \end{figure}

\begin{conj}{(Oriented Alon-Saks-Seymour Conjecture)}
There exists a polynomial $P$ such that for every $G$, $\chi(G)\leq P(\bpor(G))$.
\end{conj}

First of all, we prove that studying $\bp_{or}(K_m)$ is deeply linked with the existence of a fooling set for $CL-IS$. Recall the definitions of Section \ref{sec: CS-sep}: in the communication matrix $M$ for $CL-IS$, each row corresponds to a clique $K$, each column corresponds to a stable set $S$, and $M_{K,S}=1$ if $K$ and $S$ intersect, 0 otherwise. A \emph{fooling set} $\mathcal{C}$ is a set of pairs $(K,S)$ such that $K$ and $S$ do not intersect, and for all $(K,S), (K', S') \in \mathcal{C}$, $K$ intersects $S'$ or $K'$ intersects $S$ (consequently $M_{K,S'}=1$ or $M_{K',S}=1$). Thus $\mathcal{C}$ is a set of 0-entries of the matrix that pairwise can not be put together into the same combinatorial 0-rectangle. The maximum size of a fooling set consequently is a lower bound on the non-deterministic communication complexity for $CL-IS$, and consequently on the size of a CS-separator.

The proofs of Theorems \ref{th:fooling set} and \ref{th equiv CS ASS} are very similar one to each other and are a very close adaptation of the work by Alon and Haviv \cite{AlonH} (described in \cite{Sudakov10}). They proved statements analogous to Lemmas \ref{lem ASS CS} and \ref{lem CS ASS} where the notion of oriented bipartite packing number  $\bpor$ is replaced with 2-biclique covering number $\bp_2$ (see Subsection \ref{subsec: t-biclique} for a definition) and bipartite packing number $\bp$, respectively. The key point is to see that $\bpor$ is the right criterion. This argument has also independently appeared in \cite{Amano}.

\begin{theorem}\label{th:fooling set}
Let $n,m \in \mathbb{N}^*$. There exists a fooling set $\mathcal{C}$ of size $m$ on some graph on $n$ vertices if and only if $\bpor(K_m)\leq n$.
\end{theorem}

\begin{lemma}\label{lemma: fooling set => bp}
Let $n,m \in \mathbb{N}^*$. If there exists a fooling set $\mathcal{C}$ of size $m$ on some graph $G$ on $n$ vertices then $\bpor(K_m)\leq n$.
\end{lemma}

\begin{proof}
Consider all pairs $(K,S)$ of cliques and stable set in the fooling set $\mathcal{C}$, and construct an auxiliary graph $H$ in the same way as in the proof of Lemma \ref{lem ASS CS}: the vertices of $H$ are the $m$ pairs $(K,S)$ of the fooling set and there is an edge between $(K,S)$ and $(K',S')$ if and only if there is a vertex in $S\cap K'$ or in $S' \cap K$. By definition of a fooling set, $H$ is a complete graph. For $x\in V(G)$, let $(A_x, B_x)$ be the oriented bipartite subgraph of $H$ where $A_x$ is the set of pairs $(K,S)$ for which $x\in K$, and $B_x$ is the set of pairs $(K,S)$ for which $x \in S$. This defines a packing certificate of size $n$ on $H$ : first of all, by definition of the edges, $(A_x, B_x)$ is complete. Moreover, every edge is covered by such an oriented bipartite subgraph: if $(K,S)(K',S')\in E(H)$ then there exists $x \in S \cap K'$ or $ x \in S' \cap K$ thus the corresponding arc is in $(A_x,B_x)$. Finally, an arc $(K,S)(K',S')$ can not appear in both $(A_x,B_x)$ and $(A_y, B_y)$ otherwise the stable set $S$ and the clique $K'$ intersect on two vertices $x$ and $y$, which is impossible. Hence $\bpor(H)\leq n$. $H$ being a complete graph on $m$ elements proves the lemma.
\end{proof}

\begin{lemma}\label{lemma: bp=>fooling set}
Let $n,m \in \mathbb{N}^*$. If $\bpor(K_m)\leq n$ then there exists a fooling set of size $m$ on some graph $G$ on $n$ vertices.
\end{lemma}

\begin{proof}
Construct an auxiliary graph $H$: the vertices are the elements of a packing certificate of size $n$, and there is an edge between $(A_1, B_1)$ and $(A_2,B_2)$ if and only if there is a vertex $x\in A_1\cap A_2$. Then for all $x\in V(K_m)$, the set of all bipartite graphs $(A,B)$ with $x\in A$ form a clique called $K_x$, and the set of all bipartite graphs $(A,B)$ with $x\in B$ form a stable set called $S_x$. $S_x$ is indeed a stable set, otherwise there are $(A_1,B_1)$ and $(A_2, B_2)$ in $S_x$ (implying $x\in B_1\cap B_2$) linked by an edge resulting from a vertex $y \in A_1\cap A_2$, then the arc $yx$ is covered twice. 
Consider all pairs $(K_x, S_x)$ for $x \in V(K_m)$: this is a fooling set of size $m$. Indeed, on one hand $K_x \cap S_x=\emptyset$. On the other hand, for all $x, y \in V(K_m)$, the edge $xy$ is covered by a complete bipartite graph $(A,B)$ with $x\in A$ and $y\in B$ (or conversely). Then $K_x$ and $S_y$ (or $K_y$ and $S_x)$ intersects in $(A,B)$.
\end{proof}

\begin{proof}[Proof of Theorem \ref{th:fooling set}]
Lemmas \ref{lemma: fooling set => bp} and \ref{lemma: bp=>fooling set} conclude the proof.
\end{proof}

One can search for an algebraic lower bound for $\bpor(K_m)$. Let $(A_1,B_1), \ldots, (A_k,B_k)$ be a packing certificate of $K_m$. For every $i$ construct the $m \times m$ matrix $M^i$ such that $M^i_{u,v}=1$ if $u \in A_i, v\in B_i$ and 0 otherwise, then $M^i$ has rank 1. Let $M=\sum_{i=1}^{k}M^i$, then by construction $M$ has rank at most $k$, and has the three following particularities: it contains only $0$ and $1$, its diagonal entries are all $0$, and for every distinct $i,j$, $M_{i,j}=1$ or $M_{j,i}=1$ (or both). This is due to the definition of a packing certificate. A natural question arising is to find a lower bound on the minimum rank of a $m\times m$ matrix respecting these three particularities. This will imply a lower bound on $\bpor(K_m)$, and thus an upper bound on the size of a fooling set.

Theorem \ref{th:fooling set} implies that if $\bp_{or}(K_n)=\O(n^{1/k})$, then there exists a fooling set of size $\Omega(n^k)$ on some graphs $G$ on $n$ vertices, thus $\Omega(n^k)$ is a lower bound on the Clique-Stable Set separation. Note that the best upper bound so far is due to Amano\cite{Amano}: $\bpor(K_n)=\O(n^{2/3})$ which implies a $\Omega(n^{3/2})$ lower bound on the Clique-Stable separation. The best lower bound for the size of a fooling set is the following:

\begin{obs}\label{th:fooling set lineaire}
Let $G$ be a graph. Then there exists a fooling set $\F$ on $G$ of size $\vert V(G) \vert +1$.
\end{obs}

\begin{proof}
Let us do the proof by induction on $\vert V(G) \vert$. If $V=\{v\}$, consider the clique $\{v\}$ together with the empty stable set, and the stable set $\{v\}$ together with the empty clique. This is a fooling set of size 2.
If $\vert V \vert =n+1$, let $v\in V$, $n_1=\vert N(v)\vert$, $n_2=\vert N^C [v]\vert$, with $n=n_1+n_2+1$. Then the induction hypothesis gives a fooling set $\F_1$ of size $n_1+1$ on $N(v)$, and a fooling set $\F_2$ of size $n_2+1$ on $N^C [v]$. Extend each clique of $\F_1$ with $v$, which still forms a clique; and extend each stable set of $\F_2$ with $v$, which still forms a stable set. This gives a fooling set $\F$ of size $n_1+1+n_2+1=n+1$. It is indeed a fooling set: if $(K,S), (K', S') \in \F$, either they come both from $\F_1$ or both from $\F_2$, so the property is verified by $\F_1$ and $\F_2$ being fooling sets; either $(K,S)$ initially comes from $\F_1$ and $(K',S')$ from $\F_2$, and then $K\cap S'=\{v\}$.
\end{proof}

In fact the oriented Alon-Saks-Seymour conjecture is equivalent to the Clique-Stable Set separation conjecture.

\begin{theorem} \label{th equiv CS ASS}
The oriented Alon-Saks-Seymour conjecture is verified if and only if the Clique-Stable Set separation conjecture is verified.
\end{theorem}

As already mentioned, the proof is very similar to the one of Theorem \ref{th:fooling set}.

\begin{lemma} \label{lem ASS CS}
 If the oriented Alon-Saks-Seymour conjecture is verified, then the Clique-Stable Set separation conjecture is verified.
\end{lemma}

\begin{proof}
Let $G$ be a graph on $n$ vertices. We want to separate all the pairs of cliques and stable sets which do not intersect. Consider all the pairs $(K,S)$ such that the clique $K$ does not intersect the stable set $S$. Construct an auxiliary graph $H$ as follows. The vertices of $H$ are the pairs $(K,S)$ and there is an edge between a pair $(K,S)$ and a pair $(K',S')$ if and only if there is a vertex $x \in S \cap K'$ or $ x \in S' \cap K$.
For every vertex $x$ of $G$, let $(A_x, B_x)$ be the oriented bipartite subgraph of $H$ where $A_x$ is the set of pairs $(K,S)$ for which $x\in K$, and $B_x$ is the set of pairs $(K,S)$ for which $x \in S$. By definition of the edges, $(A_x, B_x)$ is complete. Moreover, every edge is covered by such an oriented bipartite subgraph: if $(K,S)(K',S')\in E(H)$ then there exists $x \in S \cap K'$ or $ x \in S' \cap K$ thus the corresponding arc is in $(A_x,B_x)$. Finally, an arc $(K,S)(K',S')$ can not appear in both $(A_x,B_x)$ and $(A_y, B_y)$ otherwise the stable set $S$ and the clique $K'$ intersect on two vertices $x$ and $y$, which is impossible.
Hence the oriented bipartite packing number of this graph is at most $n$. \newline
If the oriented Alon-Saks-Seymour conjecture is verified, $\chi(H) \leq P(n)$. Consider a color of this polynomial coloring. Let $A$ be the set of vertices of this color, so $A$ is a stable set. Then the union of all the second components (corresponding to stable sets of $G$) of the vertices of $A$ do not intersect the union of all the first components (corresponding to cliques of $G$) of $A$. Otherwise, there are two vertices $(K,S)$ and $(K',S')$ of $A$ such that $K$ intersects $S'$, thus $(K,S)(K',S')$ is an edge. This is impossible since $A$ is a stable set.

The union of the cliques of $A$ and the union of the stable sets  of $A$ do not intersect, hence it defines a cut which separates all the pairs of $A$. The same can be done for every color. Then we can separate all the pairs $(K,S)$ by $\chi(H)\leq P(n)$ cuts, which achieves the proof. 
\end{proof}

\begin{lemma}\label{lem CS ASS}
 If the Clique-Stable Set separation conjecture is verified, then the oriented Alon-Saks-Seymour conjecture is verified.
\end{lemma}

\begin{proof}
Let $G=(V,E)$ be a graph with $\bpor(G)=k$.
Construct an auxiliary graph $H$ as follows. The vertices are the elements of a packing certificate of size $k$. There is an edge between two elements $(A_1,B_1)$ and $(A_2,B_2)$ if and only if there is a vertex $x \in A_1 \cap A_2$. Hence the set of all $(A_i, B_i)$ such that $x\in A_i$ is a clique of $H$ (say the clique $K_x$ \emph{associated to $x$}). The set of all $(A_i, B_i)$ such that $y\in B_i$ is a stable set in $H$ (say the stable set $S_y$ \emph{associated to $y$}). Indeed, if $y \in B_1 \cap B_2$ and there is an edge resulting from $x \in A_1 \cap A_2$, then the arc $xy$ is covered twice which is impossible. Note that a clique or a stable set associated to a vertex can be empty, but this does not trigger any problem. Since the Clique-Stable set separation conjecture is satisfied, there are $P(k)$ (with $P$ a polynomial) cuts which separate all the pairs $(K,S)$, in particular which separate all the pairs $(K_x, S_x)$ for $x\in V$. 

Associate to each cut a color, and let us now color the vertices of $G$ with them. We color each vertex $x$ by the color of the cut separating $(K_x,S_x)$. Let us finally prove that this coloring is proper. Assume there is an edge $xy$ such that $x$ and $y$ are given the same color.  Then there exists a bipartite graph $(A,B)$ that covers the edge $xy$, hence $(A,B)$ is in both $K_x$ and $S_y$. Since $x$ and $y$ are given the same color, then the corresponding cut separates both $K_x$ from $S_x$ and $K_y$ from $S_y$. This is impossible because $K_x$ and $S_y$ intersects in $(A,B)$. Then we have a coloring with at most $P(k)$ colors.
\end{proof}

\begin{proof}[Proof of Theorem \ref{th equiv CS ASS}] This is straightforward using Lemmas \ref{lem ASS CS} and \ref{lem CS ASS}.
\end{proof}

\subsection{Generalization: $t$-biclique covering numbers} \label{subsec: t-biclique}

We introduce here a natural generalization of the Alon-Saks-Seymour conjecture, studied by Huang and Sudakov in \cite{Sudakov10}. While the Alon-Saks-Seymour conjecture deals with partitioning the edges, we relax here to a covering of the edges by complete bipartite graphs, meaning that an edge can be covered several times. Formally, a \emph{$t$-biclique covering} of an undirected graph $G$ is a collection of complete bipartite graphs that covers every edge of $G$ at least once and at most $t$ times. The
minimum size of such a covering is called the \emph{$t$-biclique covering number}, and is denoted by $\bp_t(G)$. In particular, $\bp_1(G)$ is the usual bipartite packing number $\bp(G)$.

In addition to being an interesting parameter to study in its own right, the $t$-biclique covering number of complete graphs is also closely related to a question in combinatorial geometry about neighborly families of boxes. It was studied by Zaks \cite{Zaks79} and then by Alon \cite{Alon97}, who proved that $\mathbb{R}^d$ has a $t$-neighborly family of $k$ standard boxes if and only if the complete graph $K_k$ has a $t$-biclique covering of size $d$ (see \cite{Sudakov10} for definitions and further details). Alon also gives asymptotic bounds for $\bp_t(K_k)$, then slighty improved by Huang and Sudakov \cite{Sudakov10} (see the work by Cioab\u{a} and Tait for further investigation \cite{CioabaTait12}):

$$(1 + \o(1))(t!/2^{t-1})^{1/t} k^{1/t} \leq \bp_t(K_k) \leq (1 + \o(1))t k^{1/t} \enspace.$$

Our results are concerned not only with $K_k$ but for every graph $G$. It is natural to ask the same question for $\bp_t(G)$ as for $\bp(G)$, namely:

\begin{conj}[Generalized Alon-Saks-Seymour conjecture of order $t$]
 There exists a polynomial $P_t$ such that for all graphs $G$, $\chi(G) \leq P_t(\bp_t(G))$. 
\end{conj}

A $t$-biclique covering is a fortiori a $t'$-biclique covering for all $t' \geq t$. Moreover, a packing certificate of size $\bpor(G)$, which covers each edge at most once in each direction can be seen as a non-oriented biclique covering which covers each edge at most twice. Hence, we have the following inequalities:

\begin{obs}\label{obs t-biclique}
For every graph $G$:
$$ \ldots \leq \bp_{t+1}(G) \leq \bp_t(G) \leq \bp_{t-1}(G) \leq \ldots \bp_2(G) \leq \bpor(G) \leq\bp_1(G) \enspace.$$

\end{obs}

Observation \ref{obs t-biclique} and bounds on $\bp_2(K_n)$ \cite{Alon97} give $\bp_{or}(K_n) \geq \bp_2(K_n)\geq \Omega(\sqrt{n})$. Then Theorem \ref{th:fooling set} ensures that the maximal size of a fooling set on a graph on $n$ vertices is $\O(n^2)$.

 \begin{theorem} \label{ASS t-biclique}
Let $t \in \mathbb{N^*}$. The generalized Alon-Saks-Seymour conjecture of order $t$ holds if and only if it holds for order 1.
\end{theorem}

\begin{proof} Assume the generalized Alon-Saks-Seymour conjecture of order $t$ holds.  Then $\chi(G)$ is bounded by a polynomial in $\bp_t(G)$ and thus, according to Observation \ref{obs t-biclique}, by a polynomial in $\bp_1(G)$. Hence the generalized Alon-Saks-Seymour of order 1 holds. 

Now we focus on the other direction, and assume that the generalized Alon-Saks-Seymour conjecture of order 1 holds. Let us prove the result by induction on $t$, initialization for $t=1$ being obvious. Let $G=(V,E)$ be a graph and let $\mathcal{B}=(B_1,...,B_k)$ be a $t$-biclique covering. Then $E$ can be partitioned into $E_t$ the set of edges that are covered exactly $t$ times in $\mathcal{B}$, and $E_{<t}$ the set of edges that are covered at most $t-1$ times in $\mathcal{B}$. Construct an auxiliary graph $H$ with the same vertex set $V$ as $G$ and with edge set $E_t$.

\begin{claim} \label{t-edges}
$\bp_1(H) \leq (2k)^t$.
\end{claim}

Since the Alon-Saks-Seymour of order 1 holds, then there exists a polynomial $P$ such that $\chi(H) \leq P((2k)^t)$. Consequently $V$ can be partitioned into $(S_1, \ldots, S_{P((2k)^t)})$ where $S_i$ is a stable set in $H$. In particular, the induced graph $G[S_i]$ contains no edge of $E_t$. Consequently $(B_1\cap S_i, \ldots, B_k \cap S_i)$ is a $(t-1)$ biclique covering of $G[S_i]$, where $B_j\cap S_i$ is the bipartite graph $B_j$ restricted to the vertices of $S_i$. Thus $\bp_{t-1}(G[S_i]) \leq k$. By induction hypothesis, the generalized Alon-Saks-Seymour of order $(t-1)$ holds, so there exists a polynomial $P_{t-1}$ such that $\chi(G[S_i]) \leq P_{t-1}(k)$. Let us now color the vertices of $G$ with at most $P((2k)^t)\cdot P_{t-1}(k)$ colors, which is a polynomial in $k$. Each vertex $v\in S_i$ is given color $(\alpha, \beta)$, where $\alpha$ is the color of $S_i$ in $H$ and $\beta$ is the color of $x$ in $G[S_i]$. This is a proper coloring of $G$, thus the generalized Alon-Saks-Seymour conjecture of order $t$ holds.
\end{proof}

\begin{proof}[Proof of Claim \ref{t-edges}]
For each $B_i$, let $(B_i^-, B_i^+)$ be its partition into a complete bipartite graph.
We number $x_1, \ldots, x_n$ the vertices of $H$.
Let $x_ix_j$ be an edge, with $i<j$, then $x_ix_j$ is covered by exactly $t$ bipartite graphs $B_{i_1}, \ldots, B_{i_t}$. 
We give to this edge the label $((B_{i_1}, \ldots, B_{i_t}), (\varepsilon_{1}, \ldots, \varepsilon_{t}))$, where $\varepsilon_l=-1$ if $x_i \in B_{i_l}^-$ (then $x_j\in B_{i_l}^+$) and $\varepsilon_l=+1$ otherwise (then $x_i \in B_{i_l}^+$ and $x_j\in B_{i_l}^-$).
For each such label $\L$ appearing in $H$, call $E_\L$ the set of edges labeled by $\L$ and define a set of edges $B_\L=E(B_{i_1}) \cap E_\L$. Observe that $B_\L$ forms a bipartite graph.
The goal is to prove that the set of every $B_\L$ is a 1-biclique covering of $H$. Since there can be at most $(2k)^t$ different labels, this will conclude the proof.

Let us first observe that each edge appears in exactly one $B_\L$ because each edge has exactly one label. Let $\L$ be a label, and let us prove that $B_\L$ is a complete bipartite graph. If $x_ix_{i'}\in B_\L$ and $x_jx_{j'}\in B_\L$, with $i<i'$ and $j<j'$ then these two edges have the same label $\L=((B_{i_1}, \ldots, B_{i_t}), (\varepsilon_{1}, \ldots, \varepsilon_{t}))$. If $\varepsilon_l=-1$ (the other case in handle symmetrically), then $x_i$ and $x_j$ are in $B_{i_l}^-$ and $x_{i'}$ and $x_{j'}$ are in $B_{i_l}^+$. As $B_{i_l}$ is a complete bipartite graph, then the edges $x_ix_{j'}$ and $x_jx_{i'}$ appear in  $E(B_{i_l})$. Thus these two edges have also the label $\L$, so they are in $B_\L$: as conclusion, $B_\L$ is a complete bipartite graph.
\end{proof}

\section{3-CCP and the stubborn problem}

\label{sec: 3-CCP and the stubborn pb}

We focus now on Constraint Satisfaction Problems, in particular 3-CCP and the stubborn problem. We observe that the historical tool used by Feder et al. \cite{FederHKM03} to tackle the stubborn problem, essentially later defined 2-list covering, is strongly connected to the Clique-Stable separation. A similar concept for 3-CCP exists and still have this strong connectivity. Consequently, this bridge between both areas can give a hope for results on the Clique-Stable separation by working on Constraint Satisfaction Problems.

The following definitions are illustrated on Fig.~\ref{fig def} and deal with list coloring. Let $G$ be a graph and $\Col$ a set of $k$ colors. A set of possible colors, called \emph{constraint}, is associated to each vertex. If the set of possible colors is $\Col$ then the constraint on this vertex is \emph{trivial}. A vertex has an \emph{$l$-constraint} if its set of possible colors has size at most $l$. An \emph{$l$-list assignment} is a function $\L: V \to \P(\Col)$ that gives each vertex an $l$-constraint. A solution $\S$ is a coloring of the vertices $\S: V \to \Col$ that respects some requirements depending on the problem. We can equivalently consider $\S$ as a partition $(A_1, \ldots, A_k)$ of the vertices of the graph with $x \in A_i$ if and only if $\S(x)=A_i$ (by abuse of notation $A_i$ denotes both the color and the set of vertices having this color). An $l$-list assignment $\L$ is \emph{compatible} with a solution $\S$ if for each vertex $x$, $\S(x)\in \L(x)$. A set of $l$-list assignment
 \emph{covers} a solution $\S$ if at least one of the $l$-list assignment
is compatible with $\S$.

\begin{figure}
\centering
\subfigure[An instance of 3-CCP]{ \hspace{0.5cm}\label{instance1}
 \includegraphics{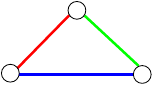}\hspace{0.5cm}}
\subfigure[A solution to the instance (vertex coloring) together with a compatible 2-list assignment: each vertex has a 2-constraint. ]{ \label{instance2}
 \hspace{1cm}\includegraphics{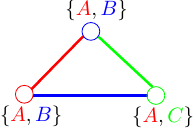}\hspace{1cm}}
\hspace{0.5cm}
\subfigure[Another solution to the instance with a compatible 2-list assignment.]{ \label{instance3}
\hspace{1cm} \includegraphics{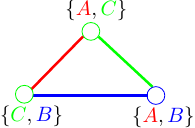}\hspace{1cm}}

\caption{Illustration of definitions. Color correspondence: \textcolor{red}{A=red} ; \textcolor{blue}{B=blue} ; \textcolor{green}{C=green}. Both 2-list assignments together form a 2-list covering because any solution is compatible with at least one of them.}
\label{fig def}
\end{figure}

We recall the definitions of 3-CCP and the stubborn problem:

\textsc{$3$-Compatible Coloring Problem
  ($3$-CCP)}
\\ \textbf{Input:} An edge coloring $f_E$ of $K_n$ with $3$ colors $\{A,B,C\}$.
\\ \textbf{Question:} Is there a coloring of the vertices with $\{A,B,C\}$, such that no edge has the same color as both its endpoints?

\medskip

\textsc{Stubborn Problem} 
\\ \textbf{Input:} A graph $G=(V,E)$ together with a list assignments $\L: V \to \mathcal{P}(\{A_1,A_2,A_3,A_4\})$.
\\ \textbf{Question:} Can $V$ be partitioned into four sets $A_1,\ldots ,A_4$ such that $A_4$ is a clique, both $A_1$ and $A_2$ are stable sets, $A_1$ and $A_3$ are completely non-adjacent, and the partition is compatible with $\L$?

%

Given an edge-coloring $f_E$ on $K_n$, a set of 2-list assignment is a \emph{2-list covering for 3-CCP on $(K_n, f_E)$} if it covers all the solutions of 3-CCP on this instance. Moreover, 3-CCP is said to have a \emph{polynomial} 2-list covering if there exists a polynomial $P$ such that for every $n$ and for every edge-coloring $f_E$, there is a 2-list covering on $(K_n, f_E)$ whose cardinality is at most $P(n)$.

\begin{obs}\label{refor}
Given a $2$-list assignment for 3-CCP, it is possible to decide in polynomial time if there exists a solution covered by it.
\end{obs}

\begin{proof}
 Any 2-list assignment can be translated into an instance of 2-SAT. Each vertex has a 2-constraint $\{\alpha, \beta\}$ from which we construct two variables $x_\alpha$ and $x_\beta$ and a clause $x_\alpha \vee x_\beta$. Turn $x_\alpha$ to true will mean that $x$ is given the color $\alpha$. Then we need also the clause $\neg x_\alpha \vee \neg x_\beta$ saying that only one color can be given to $x$. Finally for all edge $xy$ colored with $\alpha$, we add the clause $\neg x_\alpha \vee \neg y_\alpha$ if both variables exists, and no clause otherwise.
\end{proof}

Therefore, given a polynomial 2-list covering, it is possible to decide in polynomial time if the instance of 3-CCP has a solution. Observe nevertheless that the existence of a polynomial 2-list covering does not imply the existence of a polynomial algorithm. Indeed, such a 2-list covering may not be computable in polynomial time.

\begin{theorem}\label{th:quasi p}
\cite{FederH06}
There exists an algorithm giving a 2-list covering of size $O(n^{\log n})$ for 3-CCP. By Observation \ref{refor}, this gives an algorithm in time $O(n^{\log n})$ which solves $3$-CCP.
\end{theorem}

Symmetrically, we want to define a \emph{2-list covering for the stubborn problem}. However, there is no hope to cover all the solutions of the stubborn problem on each instance with a polynomial number of 2-list assignments. Indeed if $G$ is a stable set of size $n$ and if every vertex has the trivial 4-constraint, then for any partition of the vertices into 3 sets $(A_1,A_2,A_3)$, there is a solution $(A_1, A_2, A_3, \emptyset)$. Since there are $3^n$ partitions into 3 sets, and since every 2-list assignment covers at most $2^n$ solutions, all solutions cannot be covered with a polynomial number of 2-list assignments. 

Thus we need a notion of maximal solutions. This notion is extracted from the notion of domination (here $A_3$ dominates $A_1$) in the language of general list-$M$ partition problem (see \cite{FederHKM03}). Intuitively, if $\L(v)$ contains both $A_1$ and $A_3$ and $v$ belongs to $A_1$ in some solution $\S$, we can build a simpler solution by putting $v $ in $ A_3$ and leaving everything else unchanged. A solution $(A_1, A_2, A_3, A_4)$ of the stubborn problem on $(G, \L)$ is a \emph{maximal solution} if no member of $A_1$ satisfies $A_3 \in \L(v)$. We may note that if $A_3$ is contained in every $\L(v)$ for $v \in V$, then every maximal solution of the stubborn problem on $(G, \L)$ let $A_1$ 
empty. Now, a set of 2-list assignments is a \emph{2-list covering for the stubborn problem on $(G, \L)$} if it covers all the maximal solutions on this instance. Moreover, it is called a \emph{polynomial} 2-list covering if its size is bounded by a polynomial in the number of vertices in $G$.

 For edge-colored
graphs, an \emph{$ (\alpha_1,...,\alpha_k)$-clique} is a clique for which every edge has a color in $\{ \alpha_1,...,\alpha_k \}$.
A \emph{split} graph is a graph in which vertices can be partitioned into an $\alpha$-clique and a $\beta$-clique. The \emph{$\alpha$-edge-neighborhood} of $x$ is the set of vertices $y$ such that $xy$ is an \emph{$\alpha$-edge}, i.e an edge colored with $\alpha$.
The \emph{majority color of $x \in V$} is the color $\alpha$ for which the $\alpha$-edge-neighborhood of $x$ is maximal in terms of cardinality (in case of ties,
we arbitrarily cut them).

This section is devoted to the proof of the following result:

\begin{theorem}\label{th equiv stubborn 3ALCP CS-sep}
 The following are equivalent:
\begin{enumerate}
 \item For every graph $G$ and every list assignment $\L: V \to \mathcal{P}(\{A_1, A_2, A_3, A_4\})$, there is a polynomial 2-list covering for the stubborn problem on $(G, \L)$.
 \item For every $n$ and every edge-coloring $f: E(K_n) \to \{A, B, C\}$, there is a polynomial 2-list covering for 3-CCP on $(K_n, f)$.
 \item For every graph $G$, there is a polynomial CS-separator.
 
\end{enumerate}
\end{theorem}

We decompose the proof into three lemmas, each of which describing one implication.

\begin{lemma}\label{lemma stubborn->3ALCP}
 $(1 \Rightarrow 2)$: Suppose for every graph $G$ and every list assignment $\L: V \to \mathcal{P}(\{A_1, \ldots, A_4\})$, there is a polynomial 2-list covering for the stubborn problem on $(G, \L)$. Then for every graph $n$ and every edge-coloring $f: E(K_n) \to \{A, B, C\}$, there is a polynomial 2-list covering for 3-CCP on $(K_n, f)$.
\end{lemma}

\begin{proof}
Let $n \in \mathbb{N}$, $(K_n, f)$ be an instance of 3-CCP, and $x$ a vertex of $K_n$. Let us build a polynomial number of 2-list assignments that cover all the solutions where $x$ is given color $A$. Since the colors are symmetric, we just have to multiply the number of $2$-list assignments by $3$ to cover all the solutions. Let $(A,B,C)$ be a solution of 3-CCP where $x \in A$.

\begin{claim}\label{really3col} 
Let $x$ be a vertex and $\alpha,\beta,\gamma$ be the three different colors. Let $U$ be the $\alpha$-edge-neighborhood of $x$. If there is a $\beta \gamma$-clique $Z$ of $U$ which is not split, then there is no solution where $x$ is colored with $\alpha$.
\end{claim}
\begin{proof}
Consider a solution in which $x$ is colored with $\alpha$. All the vertices of $Z$ are of color $\beta$ or $\gamma$ because they are in the $\alpha$-edge-neighborhood of $x$. The vertices of $Z$ colored with $\beta$ form a $\gamma$-clique, those colored by $\gamma$ form a $\beta$-clique. Hence $Z$ is split.
\end{proof}
A vertex $x$ is \emph{really $3$-colorable} if for each color $\alpha$, every $\beta \gamma$-clique of the $\alpha$-edge-neighborhood of $x$ is a split graph. If a vertex is not really $3$-colorable then, in a solution, it can be colored by at most $2$ different colors. Hence if $K_n[V \backslash x]$ has a polynomial $2$-list covering, the same holds for $K_n$ by assigning the only two possible colors to $x$ in each 2-list assignment.

\medskip

\begin{figure}
 \centering $
\begin{array}{|c|c||c|}
 \hline 
 f(v) & f'(v) & f''(v)\\ \hline 
 A_2 \textrm{ or } A_1,A_2 & \ast & C \\ \hline 
 A_3 \textrm{ or } A_1,A_3 & \ast & B,C \\ \hline 
 A_4 \textrm{ or } A_1,A_4 & \ast & A \\ \hline 
 A_2,A_4 & \ast & A,C \\ \hline 
 A_2,A_3 & \ast & B,C \\ \hline
 A_3, A_4 & A'_2 \textrm{ or } A'_1,A'_2 & B \\ \hline 
 A_3, A_4 & A'_3 \textrm{ or } A'_1,A'_3 & A,C \\ \hline 
 A_3, A_4 & A'_4 \textrm{ or } A'_1,A'_4 & C \\ \hline 
 A_3, A_4 & A'_2,A'_4 & B,C \\ \hline
 A_3, A_4 & A'_2,A'_3 & A,B \\ \hline
 A_3, A_4 & A'_3,A'_4 & A,C \\ \hline

 \end{array}$
 
 \caption{This table describes the rules used in proof of lemma \ref{lemma stubborn->3ALCP} to built a 2-list assignment $f''$ for 3-CCP from a pair $(f, f')$ of 2-list assignment for two instances of the stubborn problem. Symbol $\ast$ stands for any constraint. For simplicity, we write $X,Y$ (resp. $X$) instead of $\{X, Y\}$ (resp. $\{X\}$).}
\label{fig tableau 2-list coloring for 3-ALCP}
\end{figure}

Thus we can assume that $x$ is really 3-colorable, otherwise there is a natural 2-constraint on it. Since we assume that the color of $x$ is $A$, we can consider that in all the following 2-list assignments, the constraint $\{B,C\}$ is given to the $A$-edge-neighborhood of $x$. Let us abuse notation and still denote by $(A,B,C)$ the partition of the $C$-edge-neighborhood of $x$, induced by the solution $(A,B,C)$. Since there exists a solution where $x$ is colored by $C$, and $C$ is a $AB$-clique, then Claim \ref{really3col} ensures that $C$ is a split graph $C'\uplus C''$ with $C'$ a $B$-clique and $C''$ a $A$-clique. The situation is described in Fig.~\ref{neighborhood1_bis}. Let $H$ be the non-colored graph with vertex set the $C$-edge-neighborhood of $x$ and with edge set the union of $B$-edges and $C$-edges (see Fig.~\ref{neighborhood2_bis}). 
Moreover, let $H'$ be the non-colored graph with vertex set the $C$-edge-neighborhood of $x$ and with edge set the $B$-edges (see Fig.~\ref{neighborhood3_bis}). We consider $(H, \L_0)$ and $(H', \L_0)$ as two instances of the stubborn problem, where $\L_0$ is the trivial list assignment that gives each vertex the constraint $\{A_1, A_2, A_3, A_4\}$. 

By assumption, there exists $\F$ (resp. $\F'$) a polynomial 2-list covering for the stubborn problem on $(H, \L_0)$ (resp. $(H', \L_0)$). We construct $\F''$ the set of 2-list assignment $f''$ built from all the pairs $(f,f') \in \F\times \F'$  according to the rules described in Fig.~\ref{fig tableau 2-list coloring for 3-ALCP} (intuition for such rules is given in the next paragraph). $\F''$ aims at being a polynomial 2-list covering for 3-CCP on the $C$-edge-neighborhood of $x$.

The following is illustrated on Fig.~\ref{neighborhood2_bis} and \ref{neighborhood3_bis}. Let $\S$ be the partition defined by $A_1=\emptyset$, $A_2=C''$, $A_3=B \cup C'$ and $A_4=A$. We can check that $A_2$ is a stable set and $A_4$ is a clique (the others restrictions are trivially satisfied by $A_1$ being empty and $\L_0$ being trivial). In parallel, let $\S'$ be the partition defined by $A_1'=\emptyset$, $A_2'=B$, $A_3'=A \cup C''$ and $A_4=C'$. We can also check that $A_2'$ is a stable set and $A_4'$ is a clique. Thus $\S$ (resp. $\S'$) is a maximal solution for the stubborn problem on $(H, \L_0)$ (resp. $(H', \L_0)$) inherited from the solution $(A,B,C=C' \uplus C'')$ for 3-CCP.

Let $f\in \F$ (resp. $f' \in \F'$) be a 2-list assignment compatible with $\S$ (resp. $\S'$). Then $f'' \in \F''$ built from $(f,f')$ is a 2-list assignment compatible with $(A,B,C)$.

Doing so for the $B$-edge-neighborhood of $x$ and pulling everything back together gives a polynomial 2-list covering for 3-CCP on $(K_n, f)$.

\begin{figure}
\centering
\subfigure[Vertex $x$, its $A$-edge-neighborhood subject to the constraint $\{B,C\}$, and its $C$-edge-neighborhood separated in different parts.]{ \label{neighborhood1_bis}
 \includegraphics[scale=1]{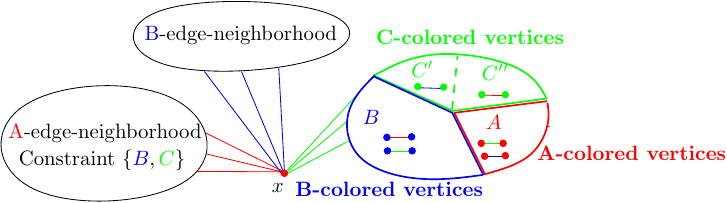}}

\subfigure[On the left, the graph $H$ obtained from the $C$-edge-neighborhood by keeping only $B$-edges and $C$-edges. On the right, the solution of the stubborn problem.]{ \label{neighborhood2_bis}
 \includegraphics[scale=1]{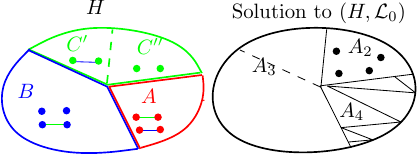}}
\subfigure[On the left, the graph $H'$ obtained from the $C$-edge-neighborhood by keeping only $B$-edges. On the right,, the solution of the stubborn problem.]{ \label{neighborhood3_bis}
 \includegraphics[scale=1]{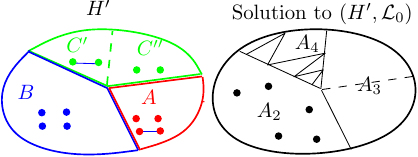}}

\caption{Illustration of the proof of lemma \ref{lemma stubborn->3ALCP}. Color correspondence: \textcolor{red}{A=red} ; \textcolor{blue}{B=blue} ; \textcolor{green}{C=green}. As before, cliques are represented by hatched sets, stable sets by dotted sets.}
\label{fig neighbourhood}
\end{figure}
\end{proof}

\begin{lemma}\label{lemma 3ALCP -> CS-sep}
 $(2 \Rightarrow 3)$: Suppose for every $n$ and every edge-coloring $f: E(K_n) \to \{A, B, C\}$, there is a polynomial 2-list covering for 3-CCP on $(K_n, f)$. Then for every graph $G$, there is a polynomial CS-separator.
\end{lemma}

\begin{proof} 
Let $G=(V,E)$ be a graph on $n$ vertices. Let $f$ be the coloring on $K_n$ defined by $f(e)=A$ if $e\in E$ and $f(e)=B$ otherwise. In the following $(K_n,f)$ is considered as a particular instance of 3-CCP with no $C$-edge. By hypothesis, there is a polynomial 2-list covering $\F$ for 3-CCP on $(K_n, f)$. Let us prove that we can derive from $\F$ a polynomial CS-separator $\mathcal{C}$. 

Let $\L \in \F$ be a $2$-list assignment. Denote by $X$ (resp. $Y$, $Z$) the set of vertices with the constraint $\{A,B\}$ (resp. $\{B,C\}$, $\{A,C\}$). Since no edge has color $C$, $X$ is split. Indeed, the vertices of color $A$ form a $B$-clique and conversely. Given a graph, there is a linear number of decompositions into a split graph \cite{FederHKM03}. Thus there are a linear number of decomposition $(U_k,V_k)_{k\leq c n}$ of $X$ into a split graph where $U_k$ is a $B$-clique. For every $k$, the cut $(U_k \cup Y, V_k \cup Z)$ is added in $\mathcal{C}$. For each $2$-list assignment we add a linear number of cuts, so the size of $\mathcal{C}$ is polynomial.

Let $K$ be a clique and $S$ a stable set of $G$ which do not intersect. The edges of $K$ are colored by $A$, and those of $S$ are colored by $B$. Then the coloring $\S(x)=B$ if $x\in K$, $\S(x)=A$ if $x \in S$ and $\S(x)=C$ otherwise is a solution of $(K_n,f)$. Left-hand side of Fig.~\ref{fig split_clique_stable} illustrates the situation. There is a 2-list assignment $\L$ in $\F$ which is compatible with this solution. As before, let $X$ (resp. $Y$, $Z$) be the set of vertices which have the constraint $\{A,B\}$ (resp. $\{B,C\}$, $\{A,C\}$). Since the vertices of $K$ are colored $B$, we have $K \subseteq X \cup Y$ (see right hand-side of Fig.~\ref{fig split_clique_stable}). Likewise, $S \subseteq X \cup Z$.
Then  $(K \cap X, S \cap X)$ forms a split partition of $X$. So, by construction, there is a cut $((K \cap X) \cup Y,(S \cap X)\cup Z) \in \mathcal{C}$ which ensures that $(K,S)$ is separated by $\mathcal{C}$.
\end{proof}

\begin{figure}
\centering
 
 \includegraphics[height=2cm]{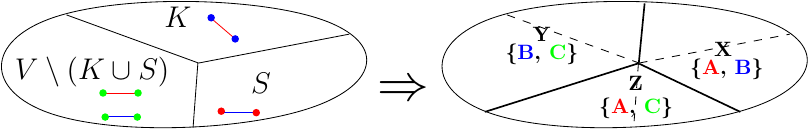}

\caption{Illustration of the proof of Lemma \ref{lemma 3ALCP -> CS-sep}. On the left hand-side, $G$ is separated in 3 parts: $K$, $S$, and the remaining vertices. Each possible configuration of edge- and vertex-coloring are represented. On the right-hand-side, $(X,Y,Z)$ is a 2-list assignment compatible with the solution. $X$ (resp. $Y$, $Z$) has constraint $\{A,B\}$ (resp. $\{B,C\}$, $\{A,C\}$). Color correspondence: \textcolor{red}{A=red} ; \textcolor{blue}{B=blue} ; \textcolor{green}{C=green}.}
\label{fig split_clique_stable}
\end{figure}

\begin{lemma} \label{lemma CS-sep -> stubborn}
 $(3 \Rightarrow 1)$: Suppose for every graph $G$, there is a polynomial CS-separator. Then for every graph $G$ and every list assignment $\L: V \to \mathcal{P}(\{A_1, A_2, A_3, A_4\})$, there is a polynomial 2-list covering for the stubborn problem on $(G, \L)$.
\end{lemma}

\begin{proof}
Let $(G, \L)$ be an instance of the stubborn problem. By assumption, there is a polynomial CS-separator for $G$. 

\begin{claim}\label{obs union stable}
If there are $p$ cuts that separate all the cliques from the stable sets, then there are $p^2$ cuts that separate all the cliques from the unions $S \cup S'$ of two stable sets. 
\end{claim}
\begin{proof}
Indeed, if $(V_1, V_2)$ separates $K$ from $S$ and $(V_1', V_2')$ separates $K$ from $S'$, then the new cut $(V_1 \cap V_1', V_2 \cup V_2')$ satisfies $K \subseteq  V_1 \cap V_1'$ and $S \cup S' \subseteq V_2 \cup V_2'$.
\end{proof}

Let $\F_2$ be a polynomial family of cuts that separate all the cliques from unions of two stable sets, which exists by Claim \ref{obs union stable} and hypothesis. Then for all $(U,W) \in \F_2$, we build the following 2-list assignment $\L'$:

\begin{enumerate}
 \item If $v \in U$, let $\L'(v)=\{A_3, A_4\}$.
 \item If $v \in W$ and $A_3 \in \L(v)$, then let $\L'(v)=\{A_2, A_3\}$.
 \item Otherwise, $v \in W$ and $A_3 \notin \L(v)$, let $\L'(v)=\{A_1, A_2\}$.
\end{enumerate}

Now the set $\F'$ of such 2-list assignment $\L'$ is a 2-list covering for the stubborn problem on $(G, \L)$: let $\S=(A_1, A_2, A_3, A_4)$ be a maximal solution of the stubborn problem on this instance. Then $A_4$ is a clique and $A_1, A_2$ are stable sets, so there is a separator $(U,W) \in \F_2$ such that $A_4 \subseteq U$ and $A_1 \cup A_2 \subseteq W$ (see Fig.~\ref{fig stubborn_cut}), and there is a corresponding 2-list assignment $\L' \in \F'$. Consequently, the 2-constraint $\L'(v)$ built from rules 1 and 3 are compatible with $\S$. Finally, as $\S$ is maximal, there is no $v \in A_1$ such that $A_3 \in \L(v)$: the 2-constraints built from rule 2 are also compatible with $\S$.
\begin{figure}
\centering
 
 \includegraphics{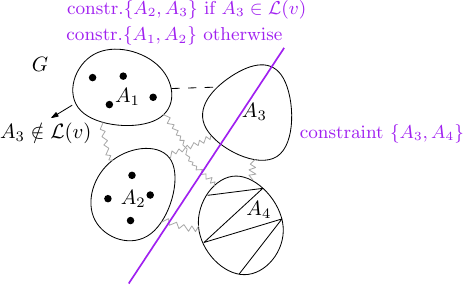}

\caption{Illustration of the proof of Lemma \ref{lemma CS-sep -> stubborn}. A solution to the stubborn problem together with the cut that separates $A_4$ from $A_1\cup A_2$. The 2-list assignment built from this cut is indicated on each side.}
\label{fig stubborn_cut}
\end{figure}
\end{proof}

\begin{proof}[Proof of theorem \ref{th equiv stubborn 3ALCP CS-sep}.] Lemmas \ref{lemma stubborn->3ALCP}, \ref{lemma 3ALCP -> CS-sep} and \ref{lemma CS-sep -> stubborn} conclude the proof of Theorem \ref{th equiv stubborn 3ALCP CS-sep}.
\end{proof}

\section*{Acknowledgments}

We thank the referees for their useful comments.

%
%
%
%
%
%
%

\bibliographystyle{plain}

\end{document}